\documentclass[a4paper,UKenglish,cleveref,autoref,thm-restate,nolineno]{socg-lipics-v2021}

\newif\ifarxive
\arxivetrue 

\ifarxive
\pdfoutput=1 
\hideLIPIcs  
\fi

\usepackage{commath}

\DeclareMathOperator{\RR}{\mathbb{R}}

\DeclareMathOperator{\QQ}{\mathbb{Q}}

\DeclareMathOperator{\ZZ}{\mathbb{Z}}

\DeclareMathOperator{\BB}{\mathbb{B}}

\DeclareMathOperator{\BC}{\mathcal{B}}

\DeclareMathOperator{\EC}{\mathcal{E}}
\DeclareMathOperator{\PC}{\mathcal{P}}

\DeclareMathOperator{\AC}{\mathcal{A}}

\DeclareMathOperator{\JC}{\mathcal{J}}
\DeclareMathOperator{\IC}{\mathcal{I}}

\DeclareMathOperator{\KS}{\mathscr{K}}

\DeclareMathOperator{\rank}{rank}

\DeclareMathOperator{\inputsize}{input\:size}
\DeclareMathOperator{\poly}{poly}
\DeclareMathOperator{\polylog}{polylog}

\DeclareMathOperator{\vol}{vol}
\DeclareMathOperator{\diag}{diag}

\DeclareMathOperator{\disc}{disc}
\DeclareMathOperator{\herdisc}{herdisc}

\DeclareMathOperator{\detlb}{detlb}

\DeclareMathOperator{\BUnit}{\mathbf 1}
\DeclareMathOperator{\BZero}{\mathbf 0}
\DeclareMathOperator{\xB}{\mathbf{x}}

\DeclareMathOperator{\supp}{supp}

\newcommand*{\intint}[2][1]{\left\{#1,\, \dots,\, #2\right\}}

{\par\color{red}} 
{} 

{\begin{quote}\color{blue}} 
{\end{quote}} 

{\par\color{green}} 
{} 

{\par\color{blue}} 
{}

{\par\color{blue}} 
{}

\DeclareMathOperator{\DP}{D\!P}

\newcommand\restr[2]{{
  \left.\kern-\nulldelimiterspace 
  #1 
  \vphantom{\big|} 
  \right|_{#2} 
  }}
\theoremstyle{definition}
\newtheorem{problem}{Problem}
\crefname{problem}{problem}{problems}
\Crefname{problem}{Problem}{Problems}
\DeclareMathOperator{\ILPSF}{ILP-SF}
\DeclareMathOperator{\ILPCF}{ILP-CF}

\title{Algorithms for Standard-form ILP Problems via Koml\'os' Discrepancy Setting (Refined $2^{O(k)}$-analysis)} 

\author{Dmitry Gribanov\footnote{Corresponding author}}{Laboratory  LATNA, HSE University, Russian Federation \and T-Technologies, Russian Federation \and \url{https://www.hse.ru/en/org/persons/51820259/}}{dimitry.gribanov@gmail.com}{https://orcid.org/0000-0002-4005-9483}
{Dmitry Gribanov is partially supported in the framework of the Basic Research Program at HSE University (HSE-BR-2025-080).}

\author{Tagir Khayaleyev\footnote{D.~Gribanov and T.~Khayaleyev are listed first, to reflect that they were the primary contributors to this work; the remaining authors are listed alphabetically.}}{Moscow Independent Research Institute of Artificial Intelligence (MIRAI), Russian Federation}{khayaleevt@gmail.com}{}{}

\author{Mikhail Cherniavskii}{
T-Technologies, Russian Federation}{cherniavskii.miu@phystech.edu}{}{}

\author{Maxim Klimenko}{
T-Technologies, Russian Federation}{klimkomx@gmail.com}{}{}

\author{Dmitry Malyshev}{Moscow Independent Research Institute of Artificial Intelligence (MIRAI), Russian Federation \and Lobachevsky State University of Nizhny Novgorod, Russian Federation \and  \url{https://www.hse.ru/en/org/persons/22927140/}}{dsmalyshev@rambler.ru}{https://orcid.org/0000-0001-7529-8233}{
Dmitry Malyshev is partially supported by the Ministry of Science and Higher Education of the Russian Federation (Goszadaniye), project No. FSMG-2024-0025.
}

\author{Stanislav Moiseev}{T-Technologies, Russian Federation}{s.moiseev@t-tech.dev}{https://orcid.org/0009-0007-0599-8999}{}

\authorrunning{D. Gribanov, T. Khayaleyev, M. Cherniavskii, M. Klimenko, D. Malyshev, S. Moiseev} 

\Copyright{Dmitry Gribanov, Tagir Khayaleyev, Mikhail Cherniavskii, Maxim Klimenko, Dmitry Malyshev, Stanislav Moiseev} 



\category{} 

\relatedversiondetails{Full Version}{https://arxiv.org/abs/2604.09806} 

\acknowledgements{The authors would like to thank Nikita Kasyanov, Alexander Rogozin and Demyan Yarmoshik for their participation in discussions and helpful comments.\\
The initials LD from the definition of \emph{LD-preconditioner} also serve as a quiet tribute to Liza D., who many times made difficult things feel worth doing [T. Khayaleyev].\\
Additionally, the authors would like to thank the anonymous reviewers for their insightful remarks and helpful suggestions, which helped us considerably improve the analysis.}

\EventEditors{Philip Bille, Seth Pettie, and Sabine Storandt}
\EventNoEds{3}
\EventLongTitle{34th Annual European Symposium on Algorithms (ESA 2026)}
\EventShortTitle{ESA 2026}
\EventAcronym{ESA}
\EventYear{2026}
\EventDate{August 31--September 4, 2026}
\EventLocation{L'Aquila, Italy}
\EventLogo{}
\SeriesVolume{388}
\ArticleNo{23}

\begin{document}

\maketitle

\begin{abstract}

    \noindent We study the standard-form ILP problem
\begin{align*}
    c^\top x \to \max\\
    \begin{cases}
        Ax = b,\\
        x \in \ZZ_{\geq 0}^n,
    \end{cases}
\end{align*} where \(A\in\ZZ^{k\times n}\) has full row rank. We obtain refined FPT
algorithms parameterized by \(k\) and \(\Delta\), the maximum absolute
value of a \(k\times k\) minor of \(A\). Our approach combines
discrepancy-based dynamic programming with matrix $\ell_2$-discrepancy bounds in
Koml\'os' setting. 
Up to polynomial factors in the input size, the optimization problem can be
solved in time \(2^{O(k)}\Delta^2\), and the corresponding
feasibility problem in time \(2^{O(k)}\Delta\). 

We would like to point out that the original version of this work, submitted to ESA 2026 and published in \cite{ESA2026_Grib}, contains a slightly weaker result. More precisely, the earlier bounds were $O(\kappa_k)^{2k} \Delta^2$ for optimization and $O(\kappa_k)^{k} \Delta$ for feasibility, where \(\kappa_k\) denotes the maximum $\ell_\infty$-discrepancy in Koml\'os' setting. However, following the advice of one of the anonymous reviewers — to whom we are very grateful — we replaced the $\ell_\infty$-discrepancy with the $\ell_2$-discrepancy in the analysis. This is what allowed us to replace dependencies of the form $\kappa_k^{O(k)}$ with $2^{O(k)}$.

\end{abstract}

\tableofcontents

\section{Introduction}\label{sec:intro}

We consider the following ILP problem:
\begin{definition}
Let $A \in \ZZ^{k\times n}$, $\rank(A) = k$, $c \in \ZZ^n$, $b \in \ZZ^{k}$. 
\emph{An ILP problem in standard form of codimension $k$} is the problem:
\begin{align}
    &c^\top x \to \max \notag\\
    &\begin{cases}
    A x = b\\
    x \in \ZZ_{\geq 0}^n.
    \end{cases}\label{ILP-SF}\tag{\(\ILPSF\)}
\end{align}
\end{definition}

We study the computational complexity of this problem with respect
to $k$ (codimension of the system $A x = b$) and the absolute values of subdeterminants  of $A$. These values are captured by the following definition:
\begin{definition}
For a matrix $A \in \ZZ^{k \times n}$ and $j \in \intint{k}$, by
\[
\Delta_j(A) = \max\left\{\abs{\det (A_{\IC \JC})} \colon \IC \subseteq \intint{k},\, \JC \subseteq \intint{n},\, \abs{\IC} = \abs{\JC} = j\right\},
\]
we denote the maximum absolute value of determinants of all the $j \times j$ submatrices of $A$, with the convention that
\[
\Delta_j(A)=0 \qquad \text{for } j>\min\{k,n\}.
\]
We also use the shorthand notation $\Delta(A) := \Delta_{\rank(A)}(A)$ for the maximum absolute value of a full-rank minor of $A$.
A matrix $A$ with $\Delta(A) \leq D$ is called \emph{$D$-modular}. 
\end{definition}

We will often use the shorthand notation $\Delta := \Delta(A)$ and $\Delta_1 := \Delta_1(A) = \norm{A}_{\max}$.
To the best of our knowledge, the study of integer linear programming problems of bounded codimension began with the classical work of Papadimitriou \cite{Papadimitriou}, which shows that the \eqref{ILP-SF} problem  can be solved in
\begin{equation}\label{eq:papa_complexity}
    O\left(n^{2k + 2} \cdot \bigl(k \cdot\max\{\Delta_1, \norm{b}_{\infty}\} \bigr)^{(k+1)(2k+1)} \right)
\end{equation}
arithmetic operations. Following the analysis in \cite[Section~18.6]{Schrijver}, this complexity bound can be rewritten in terms of subdeterminants of $A_{\text{ext}} = \bigl(A\; b\bigr)$:
\begin{equation}\label{eq:papa_complexity_reduced}
    n^{O(k)} \cdot D^{O(k)}.
\end{equation}
Here, $D$ denotes the maximum absolute value of subdeterminants of the extended matrix $A_{\text{ext}}$ of all orders up to $k$.

Four decades later, Eisenbrand \& Weismantel \cite{SteinitzILP} improved this bound by showing that \eqref{ILP-SF} and its feasibility variant can be solved in time
\begin{gather}
    n \cdot O(k \Delta_1)^{2k} \cdot \norm{b}_1^2,\label{eq:steinitz_complexity}\\
    n \cdot O(k \Delta_1)^{k} \cdot \norm{b}_1,\label{eq:steinitz_feasibility_complexity}
\end{gather}
respectively. Jansen \& Rohwedder \cite{DiscConvILP} later sharpened these bounds and removed dependence on $\norm{b}$, obtaining
\begin{gather}
    O(k)^k \cdot \Delta_1^{2k}/2^{\Omega(\sqrt{\log \Delta_1})} + T_{\text{LP}},\label{eq:JR_complexity}\\
    O(k)^{k/2} \cdot \Delta_1^{k} \cdot (\log \Delta_1)^2 + T_{\text{LP}}, \label{eq:JR_feasibility_complexity}
\end{gather}
for \eqref{ILP-SF} and its feasibility variant, respectively.
Here $T_{\text{LP}}$ denotes the computational complexity of solving the relaxed LP problem. It was also shown that these complexity bounds are conditionally optimal with respect to $\Delta_1$.

It is natural to investigate the computational complexity of the \eqref{ILP-SF} problem with respect to subdeterminants of orders other than $1$ in more detail, thereby extending Papadimitriou's bound in \eqref{eq:papa_complexity_reduced}. In this direction, Gribanov et al. \cite[Corollary~9]{OnCanonicalProblems_Grib} showed that the \eqref{ILP-SF} problem can be solved using  
\begin{equation}\label{eq:canoncail_work_complexity}
    O(\log k)^{2k^2 + o(k^2)} \cdot \Delta^2 \cdot \log^2(\Delta) + T_{\text{LP}}
\end{equation}
operations. However, the complexity dependence on $k$ appears overly pessimistic.

Computational complexity bounds with significantly better dependence on $k$ were obtained by Cherniavskii et al. \cite{DiscConvILP_Extension} for the \eqref{ILP-SF} problem and its feasibility variant:
\begin{gather}
    O(\log k)^{2k} \cdot \Delta^2/2^{\Omega(\sqrt{\log \Delta})} + 2^{O(k)} \cdot \poly(\inputsize),\label{eq:CGMP_complexity}\\
    O(\log k)^{k} \cdot \Delta \cdot (\log \Delta)^2 + 2^{O(k)} \cdot \poly(\inputsize).\label{eq:CGMP_feasibility_complexity}
\end{gather}
The computational complexity bounds \eqref{eq:CGMP_complexity} and \eqref{eq:CGMP_feasibility_complexity} are also applicable to problems in canonical form with $(n+k)$ inequalities (see \ref{ILP-CF} and \Cref{sec:canonical_ILP}).

\begin{remark}[Conditional Lower Bounds]
    Note that applying Hadamard's inequality to these complexity bounds yields a computational cost that almost matches \eqref{eq:JR_complexity} and \eqref{eq:JR_feasibility_complexity}. The latter fact implies that these complexity bounds are conditionally optimal with respect to $\Delta$ (for the discussion, see Section ``Conditional Lower Bounds'' of \cite{DiscConvILP_Extension}). 

    Additionally,  we  note  that,  due  to Bock et al. \cite{StableSetHardness}, there  are  no  polynomial-time  algorithms  for  the ILP problems with $\Delta = \Omega(n^\varepsilon)$, for any $\varepsilon > 0$, unless  the ETH (the Exponential Time Hypothesis) is false. 
    
    The last fact is the reason why we need to use both parameters $k$ and $\Delta$. Indeed, due to \cite{StableSetHardness}, the complexity bound $\poly(\Delta, \inputsize)$ is unlikely to exist, while the bounds of the form $2^{O(k)} \cdot \poly(\Delta, \inputsize)$ are provided in our work.
\end{remark}

\subsection{Our Contribution}

Our main result is an algorithm for $\Delta$-modular instances with $2^{O(k)}$ dependence on $k$, at the cost of only a negligible deterioration in the dependence on $\Delta$ with respect to \eqref{eq:CGMP_complexity}. 

\begin{restatable}{theorem}{ThKomlosILP}\label{th:KomlosILP}
    The \eqref{ILP-SF} problem can be solved in
    \begin{equation*}
        C^{2k + o(k)} \cdot \Delta^2 \cdot \log \Delta + \poly(\inputsize)
    \end{equation*}
    arithmetic operations, where $C = 24 e \sqrt{2 \pi} \thickapprox 163.529$. Its feasibility variant can be solved in
    \begin{equation*}
        C^{k + o(k)} \cdot \Delta \cdot \poly(\inputsize)
    \end{equation*}
    bit operations\footnote{Note that for the feasibility problem we estimate the bit complexity, because it is easier to do so for several reasons. The arithmetic complexity can be analyzed in further work.}.
\end{restatable}
\noindent The proof is given in \Cref{sec:proof:KomlosILP}. 

Recall that a matrix is said to lie in Koml\'os' setting if the Euclidean norm of each of its columns is at most one. Our algorithm builds on three well-known theories or lines of research: matrix discrepancy in Komlós' setting, a modified version of the dynamic programming scheme due to Jensen \& Rohwedder, and tools from the geometry of numbers for enumerating integer points in convex bodies. These three ideas are unified by a new concept that we introduce, the notion of an \emph{LD-preconditioner} --- a special multiplicative matrix factor that improves the discrepancy properties of a given matrix by bringing it into Koml\'os' setting, and simultaneously enables efficient enumeration of integer points in the convex body associated with it. For the optimized feasibility algorithm, we use fast sparse Boolean convolution algorithms. The manuscript is organized so that these ideas are unfolded gradually in the relevant sections.

\begin{remark}[Randomized vs.\ deterministic]
    Strictly speaking, the feasibility part of our algorithm is randomized because of the Boolean convolution subroutine from \Cref{sec:bool_conv}. More precisely, in \Cref{lm:bool_conv} we use the randomized sparse polynomial multiplication algorithm of Giorgi et al.~\cite{EssentiallyOptSparsePolyMult}. All remaining ingredients of the feasibility algorithm are deterministic.

    This source of randomness is not essential. As explained in \Cref{rem:bool_conv_det} from Appendix, one may replace that subroutine by the deterministic sparse nonnegative convolution algorithm of Bringmann, Fischer, and Nakos~\cite[Theorem~1]{DetLasVegasSparseNonnegativeConv}. This yields a deterministic implementation of the Boolean convolution routine and therefore a deterministic version of the feasibility algorithm as well. The asymptotic form of the feasibility bounds in \Cref{th:KomlosILP} remains unchanged.
\end{remark}

\begin{remark}[The distinction between this improved work and its original version]
    We would like to point out that the original version of this work, submitted to ESA 2026 and published in \cite{ESA2026_Grib}, contains a slightly weaker result. More precisely, the earlier bounds were $O(\kappa_k)^{2k} \Delta^2$ for optimization and $O(\kappa_k)^{k} \Delta$ for feasibility, where \(\kappa_k\) denotes the maximum $\ell_\infty$-discrepancy in Koml\'os' setting. The well-known Koml\'os conjecture asserts that $\kappa_k = O(1)$. Consequently, if the conjecture were true, it would imply the existence of algorithms running in time $2^{O(k)} \Delta^2$ and $2^{O(k)} \Delta$, respectively. However, following the advice of one of the anonymous reviewers — to whom we are very grateful — we replaced the $\ell_\infty$-discrepancy with the $\ell_2$-discrepancy in the analysis. This is what allowed us to replace dependencies of the form $\kappa_k^{O(k)}$ with $2^{O(k)}$.

    We also note a slight terminological change between the works. In \cite{ESA2026_Grib}, the notion of a \emph{normalized LD-preconditioner} was introduced; in the present work, it is renamed to \emph{regular LD-preconditioner}. We found the word ``regular'' to be more appropriate. In addition, in this version we have introduced the notion of the signature of an LD-preconditioner, which allows us to estimate the associated numerical parameters more accurately. This was necessary for a more precise determination of the exponential order in $2^{O(k)}$.
\end{remark}

\ifarxive

\begin{remark}[Key innovations with respect to previous works in the series]
    Our work is the third in a series of papers devoted to obtaining complexity results for ILP problems parameterized by $(k,\Delta)$. This series builds on the work of Jensen \& Rohwedder~\cite{DiscConvILP}, which developed a dynamic programming approach based on discrepancy properties of matrices and fast convolution algorithms, and which used the pair $(k,\Delta_1)$ as the main parameters. 
    
    In our first work in the series~\cite{OnCanonicalProblems_Grib}, we already employed a fairly weak version of preconditioning and Spencer's discrepancy bound \cite{SixDeviations_Spencer}. However, the main goal of~\cite{OnCanonicalProblems_Grib} was to relate problems in canonical and standard forms, showing that parametric complexity results for ILP in canonical form can also be parameterized by the pair $(k,\Delta)$ (see \Cref{sec:canonical_ILP} for the discussion).

    Our second work~\cite{DiscConvILP_Extension} substantially advanced the ideas of~\cite{OnCanonicalProblems_Grib} by using a more refined discrepancy bound depending on the parameter $\detlb(A)$. To exploit this bound, they proposed preconditioning by the largest minor of $A$, or by an approximation thereof. Moreover, \cite{DiscConvILP_Extension} introduced a powerful convolution tool for finite Abelian groups, which is particularly important for ILP in canonical form.

    In the present work, we develop a considerably more sophisticated approach to constructing a preconditioner for the original matrix $A$, which allows us to use stronger discrepancy bounds, such as those in Koml\'os' setting. In an extended version of this work, we plan to incorporate the techniques of \cite{DiscConvILP_Extension} in order to cover the canonical form as well and to improve the optimization complexity from $\Delta^2$ to $\Delta^2 / 2^{\Omega(\sqrt{\Delta})}$.
\end{remark}

\fi

\subsection{ILP Problems in Canonical Form}\label{sec:canonical_ILP}

It is also natural to consider the ILP problem in canonical form with $n$ variables and $n+k$ constraints.
\begin{definition}
Let $A \in \ZZ^{(n + k)\times n}$, $\rank(A) = n$, $c \in \ZZ^n$, $b \in \ZZ^{n+k}$. \emph{The ILP problem in canonical form with $n+k$ constraints} is defined as follows:
\begin{align}
    &c^\top x \to \max \notag\\
    &\begin{cases}
    A x \leq b\\
    x \in \ZZ^n.
    \end{cases}\label{ILP-CF}\tag{\(\ILPCF\)}
\end{align}
\end{definition}
Despite the fact that the \eqref{ILP-SF} and \eqref{ILP-CF} problems are mutually transformable, all known trivial transformations alter at least one of the key parameters $k$, $\Delta$ or the dimension of the corresponding polyhedra. The existence of a parameter-preserving transformation is a nontrivial question that was resolved positively by Gribanov et al. \cite[Lemmas 4 and 5]{OnCanonicalProblems_Grib}.

As already noted, the complexity bounds \eqref{eq:CGMP_complexity} and \eqref{eq:CGMP_feasibility_complexity} are also applicable to the \eqref{ILP-CF} problem and its feasibility version. This relies precisely on the mentioned transformation between the canonical-form problem \eqref{ILP-CF} and the standard-form problem \eqref{ILP-SF}.


The approach described in this manuscript can also be applied to the \eqref{ILP-CF} problem, which yields a computational complexity bound with a better dependence on $k$ and a negligible deterioration in $\Delta$.
\begin{theorem}\label{th:KomlosILP_feasibility}
    The \eqref{ILP-CF} problem can be solved in
    \begin{equation*}
        C^{2k + o(k)} \cdot \Delta^2 \cdot \log \Delta + \poly(\inputsize)
    \end{equation*}
    arithmetic operations, where $C = 24 e \sqrt{2 \pi} \thickapprox 163.529$.
\end{theorem}

We do not present an accelerated algorithm for the feasibility version of \eqref{ILP-CF}, since this would require a more complicated generalized Boolean convolution that is not needed for \eqref{ILP-SF}. Since \eqref{ILP-CF} plays only a secondary role in this paper, we also omit the proof of \Cref{th:KomlosILP_feasibility}. That proof can be obtained by a minor adaptation of the dynamic program used in \Cref{th:KomlosILP}, together with the approach of \cite{DiscConvILP_Extension} and the reduction of \cite[Lemmas 4 and 5]{OnCanonicalProblems_Grib}. We plan to resolve these convolution issues and develop algorithms for the \eqref{ILP-CF} problem and its feasibility variant in an extended version of our work.

\subsection{Complexity Model Assumptions}\label{assumptions_subs}
Some of the algorithms considered in this paper are analyzed in the Word-RAM model. In other words, we assume that additions, subtractions, multiplications, and divisions of rational numbers of the prescribed size (the word size) can be performed in $O(1)$ time. We take the word size to be bounded by a fixed polynomial in $\lceil \log n \rceil + k + \lceil \log \alpha \rceil$, where $\alpha$ is the maximum absolute value of elements of $A$, $b$, and $c$ in the problem formulations. 

However, sometimes we will explicitly describe the computational complexity of algorithms in terms of the number of bit operations. Depending on which complexity measure we use in our work, we will state this explicitly. That is, we will use the terms "arithmetic operations" (sometimes shortened to "operations") and "bit operations".

\section{Preliminaries}\label{sec:prelim}

For $A \in \RR^{k \times n}$, we denote by:
\begin{itemize}
    \item $A_{i\,j}$ the $(i,j)$-th entry of $A$;
    \item $A_{i\,*}$ its $i$-th row vector;
    \item $A_{*\,j}$ its $j$-th column vector;
    \item $A_{\IC \JC}$ the submatrix consisting of rows and columns of $A$, indexed by $\IC$ and $\JC$, respectively;
    \item Replacing $\IC$ or $\JC$ with $*$ selects all rows or columns, respectively;
    \item When unambiguous, we abbreviate $A_{\IC*}$ as $A_{\IC}$ and $A_{*\JC}$ as $A_{\JC}$.
\end{itemize}

\subsection{Hermite Normal Form}\label{sec:HNF}

Let $A \in \QQ^{n \times n}$ be a rational matrix of rank $n$. It is well known (see, for example, \cite{Schrijver,HNFOptAlg,StorjohannDoc}) that there exists a unimodular matrix $U \in \ZZ^{n \times n}$, such that $A = U H$, where $H \in \QQ^{n \times n}$ is upper triangular and satisfies $0 \le H_{ij} < H_{ii}$ for every $i\in\intint n$ and $j\in\intint{i-1}$. The matrix $H$ is called the \emph{Hermite Normal Form} (or, shortly, HNF) of the matrix $A$. Near-optimal polynomial-time algorithms for constructing the HNF of $A$ can be found in \cite{HNFOptAlg,StorjohannDoc}. 


\subsection{Preliminaries from Discrepancy Theory}\label{sec:prediscrep}

As noted in the Introduction, we employ tools from discrepancy theory. Below we recall the results and definitions we need.
\begin{definition}\label{def:disc}
For a matrix $A \in \RR^{k \times n}$, \emph{its discrepancy and its hereditary discrepancy} with respect to $\ell_p$-norm are defined by
\begin{gather*}
\disc_p(A) = \min_{z \in \{0,\, 1\}^n} \left\| A (z - 1/2 \cdot \BUnit)  \right\|_p,\\
\herdisc_p(A) = \max_{\IC \subset \intint n} \disc_p(A_{* \JC}).
\end{gather*}
For an arbitrary norm $\norm{\cdot}$, we adopt the analogous notation 
\begin{equation*}
    \disc_{\norm{\cdot}}(A) \quad\text{and}\quad \herdisc_{\norm{\cdot}}(A).
\end{equation*}
\end{definition}

\begin{remark}
    Note that the definition of $\disc_p$ given here differs slightly from the classical one:
    \[
        \min_{z \in \{-1,\, 1\}^n} \left\| A z \right\|_p.
    \]
    These two definitions differ only by a constant factor of $2$:
    \[
        \min_{z \in \{-1,\, 1\}^n} \left\| A z \right\|_p
        =
        2 \cdot \min_{z \in \{0,\, 1\}^n} \left\| A (z - 1/2 \cdot \BUnit) \right\|_p.
    \]
    We adopt the former \Cref{def:disc}, since it was historically chosen in the work of Jansen and Rohwedder \cite{DiscConvILP}, on which our approach is based; in this way, our dynamic program will resemble the one in \cite{DiscConvILP} as closely as possible.
\end{remark}

A central setting in discrepancy theory is the Koml\'os setting, which deals with matrices $A \in \RR^{k \times n}$ all of whose columns have Euclidean norm at most one. We denote this class of matrices by
\begin{equation*}
    \KS_n = \left\{A \in \RR^{k \times n} \colon \norm{A_{* i}}_2 \leq 1,\; \forall i \in \intint n \right\}.
\end{equation*}
The long-standing Koml\'os conjecture states that 
\begin{equation*}
    \forall A \in \KS_n,\quad \disc_\infty(A) = O(1).
\end{equation*}
Considerable work has gone into proving or disproving this conjecture, and the current best bound is due to Bansal \& Jiang \cite{BansalDecoupling}:
\begin{equation*}\label{eq:BanDiscUB}
    \forall A \in \KS_n,\quad \disc_\infty(A) = \sqrt[4]{\log n} \cdot \bigl(\log\log n\bigr)^{O(1)}.
\end{equation*}

However, as suggested by an anonymous reviewer, it is much more convenient for our purposes to use $\disc_2$ instead of $\disc_\infty$. The following classical fact, which serves as a basic illustration of the probabilistic method, states that
\begin{equation}\label{eq:Disc2UB}
    \forall A \in \KS_n,\quad \disc_2(A) \leq \frac{\sqrt{n}}{2}.
\end{equation}
Naturally, this fact can be found in Alon \& Spencer \cite[Theorem~2.4.1]{AlonSpencerBook}. We must also recall an important property concerning the discrepancy of matrices $A \in \RR^{k \times n}$ with $k \leq n$.
\begin{lemma}[{Alon \& Spencer \cite[Corollary 13.3.3]{AlonSpencerBook}}]\label{lm:DiscLowk}
    Let $\norm{\cdot}$ be an arbitrary vector norm. Suppose that $\disc_{\norm{\cdot}}(A_{* \JC}) \leq H$, for every subset $\JC \subset \intint n$ with $\abs{\JC} \leq k$. Then, $\disc(A)_{\norm{\cdot}} \leq 2 H$.
\end{lemma}
Originally, this statement was proposed only for hypergraph discrepancy with respect to $\ell_\infty$ norm. However, it is straightforward to see from the original proof that it extends to matrices and general norms as well. Combining Lemma \ref{lm:DiscLowk} with
\eqref{eq:Disc2UB}, 
we obtain
\begin{equation}\label{eq:BanDiscUBReduced}
    \forall A \in \KS_n \text{with at most $k$ rows},\quad \herdisc_2(A) \leq \sqrt{k}.
\end{equation}

The following key lemma, due to K.~Jansen \& L.~Rohwedder \cite{DiscConvILP}, connects results from discrepancy theory with ILP theory. The original version of this lemma is stated with respect to $\ell_\infty$ norm. However, it can be directly checked that it can be generalized to arbitrary norms.
\begin{lemma}[K.~Jansen \& L.~Rohwedder \cite{DiscConvILP}]\label{lm:disc_rounding} Let $A \in \RR^{k \times n}$ with $\rank(A) = k$, $\norm{\cdot}$ be an arbitrary norm and $x \in \ZZ^n_{\geq 0}$. Then, there exists a vector $z \in \ZZ^n_{\geq 0}$ with
\begin{enumerate}
    \item $z \leq x$;
    \item $
\norm{A(z - x/2)} \leq \herdisc_{\norm{\cdot}}(A)
$.
\end{enumerate}
Assuming additionally that $\norm{x}_1 > 1$, there exists a vector $z \in \ZZ^n_{\geq 0}$ with
\begin{enumerate}
    \item $z \leq x$;
    \item $\frac{1}{6} \cdot \norm{x}_1 \leq \norm{z}_1 \leq \frac{5}{6} \cdot \norm{x}_1$;
    \item $
\norm{A(z - x/2)} \leq 2 \herdisc_{\norm{\cdot}}(A)
$.
\end{enumerate}
\end{lemma}

\subsection{Enumeration of Integer Points in Rational Parallelepipeds}\label{sec:enum}

The dynamic programming algorithm we will use requires enumerating integer points of a parallelepiped as a subroutine.
    
\begin{restatable}{lemma}{ParQEnumLm}\label{lm:ParQEnum}
    For a given upper triangular matrix $H \in \QQ^{n \times n}$ with strictly positive diagonal elements and $r \in \QQ^n$, denote $\PC = r + H \cdot [0,1)^n$. The following propositions hold:
    \begin{enumerate}
        \item \begin{equation*}
            \prod\limits_{i = 1}^n \bigl\lfloor H_{i i} \bigr\rfloor \leq \abs{\PC \cap \ZZ^n} \leq \prod\limits_{i = 1}^n \bigl\lceil H_{i i} \bigr\rceil;
        \end{equation*} 
    
        \item The elements of $\PC \cap \ZZ^n$ can be enumerated using $O(n)^{2} \cdot \prod_{i = 1}^n \lceil H_{i i} \rceil$ operations.
    \end{enumerate}
\end{restatable}
\noindent Due to its length, the proof has been moved to \Cref{lm:proof:ParQEnum}.

\begin{corollary}\label{cor:ParQEnum}
    For a given upper triangular matrix $H \in \QQ^{n \times n}$ and $r \in \QQ^n$, denote $\PC = r + H \cdot [0,1)^n$ and $\Delta = \abs{\det H}$. Assume additionally that $H_{i i} \geq  \gamma$ for each $i$. Then, the following propositions hold:
    \begin{enumerate}
        \item \begin{equation*}
            \abs{\PC \cap \ZZ^n} \leq \left(1 + \frac{1}{\gamma}\right)^n \cdot \Delta;
        \end{equation*} 
    
        \item The elements of $\PC \cap \ZZ^n$ can be enumerated using $O(n^2) \cdot \left(1 + 1/\gamma\right)^n \cdot \Delta$ operations.
    \end{enumerate}
\end{corollary}

    

\subsection{Sparse Boolean Convolution via Sparse Polynomial Multiplication}\label{sec:bool_conv}

In this subsection, we study the sparse generalized Boolean convolution problem, which will be used to speed up our dynamic programming algorithm for solving feasibility problems.
\begin{definition}\label{def:bool_conv}
    Let $\alpha$ and $\beta$ be given functions from $\ZZ^k$ to $\{0,1\}$ with finite support. We assume that the functions are represented by lists of elements of their supports. 
    The \emph{Boolean convolution} of the functions $\alpha$ and $\beta$ is the function $\gamma \colon \ZZ^k \to \{0,1\}$, defined by
    \begin{equation*}
        \gamma_x = \bigvee\limits_{\substack{y \in \supp(\alpha)\\x-y \in \supp(\beta)}} \alpha_y \wedge \beta_{x-y}.
    \end{equation*}
    Note that $\supp(\gamma) = \left\{x \in \ZZ^k \colon \exists y \in \supp(\alpha) \text{ such that }(x-y) \in \supp(\beta)\right\}$.
\end{definition}
\begin{problem}[Sparse generalized Boolean convolution]\label{prob:bool_conv}
     Given two finite-support functions
    $\alpha,\beta \colon \ZZ^k \to \{0,1\}$, each represented by the list of
    points where it is equal to $1$, compute their Boolean convolution
    $\gamma$ from \Cref{def:bool_conv} in the same sparse representation.
\end{problem}
To construct a fast algorithm for  \Cref{prob:bool_conv}, we will use the recent result on sparse polynomial multiplication by Giorgi et al. \cite{EssentiallyOptSparsePolyMult}. For a multivariate polynomial $f \in \ZZ[x_1,\dots,x_n]$, let $\norm{f}_{\infty}$ be the maximum coefficient of $f$, $\#f$ be the number of non-zero coefficients of $f$, and $\deg_{x_i} f$ be the maximum degree of the variable $x_i$ in monomials of $f$.
\begin{theorem}[{Giorgi et al. \cite[Corollary 4.8]{EssentiallyOptSparsePolyMult}}]\label{th:poly_mult}
    There exists an algorithm that takes as inputs $f,g \in \ZZ[x_1, \dots, x_n]$ and $0 < \varepsilon < 1$, and computes $f g$ with probability at least $1-\varepsilon$ using 
    \begin{equation*}
        O\bigl(T \cdot (n \log d + \log C) \cdot \polylog(1/\varepsilon)\bigr)\quad \text{bit operations},
    \end{equation*}
    where $T = \max\left\{\#(fg),\#f,\#g\right\}$, $d = \max_i \deg_{x_i} (fg)$, and $C = \max\{\norm{f}_\infty,\norm{g}_\infty\}$.
\end{theorem}

The following lemma provides a simple reduction of  \Cref{prob:bool_conv} to sparse multivariate polynomial multiplication. Due to its technical nature and space limitations, we have moved the proof to \Cref{lm:proof:bool_conv}.
\begin{restatable}{lemma}{BoolConvLm}\label{lm:bool_conv}
    Problem~\ref{def:bool_conv} can be solved in
    \begin{equation*}
        O\left(T \cdot (k \log L)\right)
    \end{equation*}
    bit operations with high probability, where $T = \max\left\{\abs{\supp(\alpha)}, \abs{\supp(\beta)}, \abs{\supp(\gamma)}\right\}$ and $L$ is the maximum $\ell_{\infty}$-norm of elements in $\supp(\alpha)$ and $\supp(\beta)$.
\end{restatable}

\subsection{The L\"owner ellipsoid for a set of points}

Let $\AC \subseteq \RR^k$ be a set of $n$ points in general position. For simplicity, we assume that $\AC$ is $\BZero$-symmetric. It is a known fact that there exists a unique minimum-volume ellipsoid $\EC$, called the L\"owner-John or L\"owner ellipsoid, that contains $\AC$. Since $\AC$ is in general position and $\BZero$-symmetric, $\EC$ is $k$-dimensional and $\BZero$-symmetric.

The L\"owner ellipsoid $\EC$ can be represented by the following pair of primal and dual convex optimization problems (see, for example, Nikolov \cite[Section~2.3]{LargestSimplex_Nikolov}). The primal program is
\begin{gather}
    - \ln \det W \to \min\notag\\
    \begin{cases}
    a^\top W a \le 1,\quad a\in\AC,\\
    W\succ \BZero .
\end{cases}\label{eq:minell_primal}\tag{Primal-LJ-Ell}
\end{gather}
The dual program is
\begin{gather}
    \ln \det\left( \sum_{i=1}^n c_i \cdot a_i a_i^\top \right) \to \max\notag\\
    \begin{cases}
        \BUnit^\top c = k\\
        c \in \RR^n_{\geq 0}.
    \end{cases}\label{eq:minell_dual}\tag{Dual-LJ-Ell}
\end{gather}
The primal problem \eqref{eq:minell_primal} is a convex minimization problem over the open domain $\{ W \colon W \succ \BZero \}$ with affine constraints. Hence, due to Slater's condition \cite[Section~5.2.3]{BoydVandenbergheConvexOptimization}, there is no duality gap between \eqref{eq:minell_primal} and \eqref{eq:minell_dual}.

The following lemma shows that the volume of $\EC$ can be bounded in terms of the maximum volume of a parallelepiped spanned by vectors in $\AC$. This fact explains the main reason for the existence of \emph{LD-preconditioners} described in \Cref{sec:precond}.
\begin{lemma}\label{lm:ell_det}
    Let $A$ be the matrix whose columns are formed by elements of $\AC$. Let $\EC = \left\{x \in \RR^k \colon x^\top W x \leq 1 \right\}$ be the L\"owner ellipsoid, i.e., the minimum-volume ellipsoid containing $\AC$. Then, $\frac{1}{\sqrt{\det W}}  \leq e^{k/2} \cdot \Delta(A)$.
\end{lemma}

\noindent For the proof of \Cref{lm:ell_det}, we will use the following technical lemma, whose proof is deferred to \Cref{lm:proof:c_opt}.
\begin{restatable}{lemma}{COptLm}\label{lm:c_opt}
    For $\alpha > 0$ and integers $n \geq k > 0$, consider the following optimization problem:
    \begin{align*}
        &\sum\limits_{\JC \in \binom{\intint n}{k}} \prod\limits_{j \in \JC} c_j \to \max\\
        &\begin{cases}
            c_1 + \dots + c_n = \alpha\\
            c \in \RR_{\geq 0}^n.
        \end{cases}
    \end{align*}
    Then, the point $c = \frac{\alpha}{n} \cdot \BUnit$ is an optimal solution of this problem.
\end{restatable}

\begin{proof}[Proof Of \Cref{lm:ell_det}]
    Consider the dual program \eqref{eq:minell_dual}. Denoting $C = \diag(c)$, we have
\begin{multline*}
    1/\det W = \det\left( \sum_{i=1}^n c_i \cdot a_i a_i^\top \right) = \det \bigl( A C A^\top \bigr) = \\ 
    \sum\limits_{\JC \in \binom{\intint n}{k}} \det \bigl( (A \sqrt{C})_{*\JC}\bigr)^2 =
    \sum\limits_{\JC \in \binom{\intint n}{k}} \det\bigl(A_{*\JC}\bigr)^2 \cdot \det C_{\JC} \leq \\
    \Delta(A)^2 \cdot \sum\limits_{\JC \in \binom{\intint n}{k}} \prod\limits_{j \in \JC} c_j \overset{\text{by Lemma \ref{lm:c_opt}}}{\leq} \\
    \Delta(A)^2 \cdot \binom{n}{k} \cdot \left( \frac{k}{n} \right)^k \leq e^k \cdot \Delta(A)^2.
\end{multline*}
\end{proof}

\ifarxive
    The following technical lemma provides an algorithm for finding an approximate L\"owner ellipsoid in a format convenient for our purposes. For completeness, we provide the proof in \Cref{lm:proof:approx_LJ_ell}.
\fi
\begin{restatable}{lemma}{ApproxLJEllLm}\label{lm:approx_LJ_ell}
    Let $\AC \subseteq \QQ^k$ be a given finite $\BZero$-symmetric set of points in general position, and $\varepsilon \in \QQ_{>0}$ be a given rational number. 
    Let \(W\in\RR^{k\times k}\) be the symmetric positive-definite matrix
    such that
    \[
        \EC=\{x\in\RR^k : x^\top W x\le 1\}
    \]
    is the L\"owner ellipsoid of \(\AC\). Then, there exists a polynomial-time algorithm that computes a matrix $C \in \QQ^{k \times k}$ such that
    \begin{enumerate}
        \item $C$ is upper triangular with positive diagonal,
        
        \item for each $a \in \AC$, $a^\top C^{\top} C a \leq 1$,

        \item $\det C \geq e^{-2 \varepsilon} \cdot \det W^{1/2}$.
    \end{enumerate}
\end{restatable}

\section{Computation of the Preconditioning Matrix}\label{sec:precond}

The main idea of our approach, which is similar in spirit to the techniques used in the discrepancy minimization paper by Dadush et al.~\cite{VBalansingInAnyNorm}, is to precondition a given matrix $A \in \RR^{k \times n}$ by choosing a matrix $B \in \RR^{k \times k}$ with small determinant such that every column of $B^{-1}A$ has Euclidean norm bounded by an absolute constant. After rescaling, the matrix $B^{-1}A$ belongs to Koml\'os' setting, which yields $\herdisc_2\bigl(B^{-1}A\bigr) \leq 2\sqrt{k}$. A natural way to construct such a matrix $B$ is via the L\"owner ellipsoid. We refer to such matrices as \emph{LD-preconditioners}, where the letters LD abbreviate L\"owner and discrepancy.
\begin{definition}[LD-preconditioner]
    Let $A = (a_1,\dots,a_n) \in \RR^{k \times n}$ be a matrix of rank $k$. A nondegenerate matrix $B \in \RR^{k \times k}$ is called an \emph{LD-preconditioner of $A$} with signature $(\alpha,\beta) \in \RR^2_{> 0}$ if the following conditions are satisfied:
    \begin{enumerate}
        \item $\abs{\det B} \leq \alpha^{k + o(k)} \cdot \Delta(A)$;
        \item $\norm{B^{-1} a_i}_2 \leq \beta$ for every $i \in \intint n$.
    \end{enumerate}
\end{definition}

\begin{proposition}\label{prop:LD_herdisc}
    Let $A \in \RR^{k \times n}$ be a matrix of rank $k$. If $B$ is an LD-preconditioner of $A$ with signature $(\alpha, \beta)$, then $\herdisc_2\bigl(B^{-1} A\bigr) \leq  2 \beta \sqrt{k}$.
\end{proposition}
\begin{proof}
    By the definition, every column of $B^{-1} A$ has Euclidean norm at most $\beta$. Hence, the matrix $\beta^{-1}(B^{-1} A)$ belongs to the Koml\'os setting.
    Applying \Cref{eq:Disc2UB}, we obtain $\herdisc_2\bigl(B^{-1} A\bigr) \leq 2 \beta \sqrt{k}$.
\end{proof}

The next proposition proves the existence of LD-preconditioners.
\begin{proposition}\label{prop:good_B}
    Let $A \in \RR^{k \times n}$ be a matrix of rank $k$. Then there exists an LD-preconditioner $B$ of $A$ with signature $\bigl(\sqrt{e}, 1\bigr)$ such that $B$ is upper triangular with positive diagonal. Moreover, for a given rational matrix \(A\), such a matrix \(B\) can be computed in polynomial time.
\end{proposition}
\begin{proof}
        Denote the columns of $A$ by $(a_1, a_2, \dots a_n)$. Let \(W\in\RR^{k\times k}\) be the symmetric positive-definite matrix
        such that
        \[
            \EC=\{x\in\RR^k : x^\top W x\le 1\}
        \]
        is the L\"owner ellipsoid of the set \(\{\pm a_i\}\). Let \(B\) be the Cholesky factor of \(W^{-1}\), so that \(B^\top B = W^{-1}\). Note that \(B\) is upper triangular with positive diagonal. By \Cref{lm:ell_det},
        \begin{equation}\label{eq:prop:good_B_1}
            \abs{\det B} = \det W^{-1/2} \leq e^{k/2} \cdot \Delta(A).
        \end{equation}
        Additionally, if $y$ is the $i$-th column of $B^{-1} A$, then
        \begin{equation}\label{eq:prop:good_B_2}
            \norm{y}_2^2 = a_i^\top B^{-\top} B^{-1} a_i \leq 1.
        \end{equation}
        
        \noindent Hence, $B$ is an LD-preconditioner of $A$ with signature $\bigl(\sqrt{e}, 1 \bigr)$. 
        
        Note that $B$ can be chosen as any upper triangular matrix with positive diagonal that satisfies \eqref{eq:prop:good_B_1} and \eqref{eq:prop:good_B_2}. Thus, for rational $A$, rather than taking \(B\) to be the exact Cholesky factor, we choose \(C\) using \Cref{lm:approx_LJ_ell} with \(\varepsilon = 1\), and set \(B = C^{-1}\); the same verification applies.
\end{proof}

    

For our purposes, we require a stronger notion of an LD-preconditioner, which we call the \emph{regular LD-preconditioner}. This stronger definition equips us with an algorithmic tool that allows us to efficiently enumerate integer points in convex bodies such as hypercubes and ellipsoids, while also preserving the low discrepancy of the preconditioned matrix $B^{-1} A$. We defer the details to \Cref{sec:enum_precond}. These enumeration algorithms play a key role in the development of our dynamic programming algorithms, which we use to prove our main \Cref{th:KomlosILP}. First, we introduce the notion of \emph{integrally regular matrices}.
\begin{definition}[Integrally Regular Matrix]\label{def:regular_matrix}
    A matrix $B$ is called \emph{$\gamma$-integrally regular} (for $\gamma > 0$) if there exists a unimodular matrix $U \in \ZZ^{k \times k}$ such that $U B$ is upper triangular with diagonal elements bounded below by $\gamma$.

    Or equivalently\footnote{The second version of this definition is equivalent to the first one for rational matrices. However, since the HNF is not defined for reals, the first generalizes the second.}\footnote{The equivalence follows from the fact that every upper triangular matrix with strictly positive diagonal can be reduced to an HNF with the same diagonal using only unimodular row operations.}, a matrix $B$ is called \emph{$\gamma$-integrally regular} if each diagonal element of the HNF of $B$ is bounded by $\gamma$ from below.
\end{definition}

\begin{definition}[Regular LD-preconditioner]\label{def:norm_LD}
    Let $A \in \RR^{k \times n}$ be a matrix of rank $k$. We say that $B$ is a \emph{regular LD-preconditioner} of $A$ with signature $(\alpha, \beta, \gamma)$ if
    \begin{itemize}
        \item $B$ is an LD-preconditioner of $A$ with signature $(\alpha,\beta)$;
        \item $B$ is a $\gamma$-integrally regular matrix.
    \end{itemize}
\end{definition}

The next lemma shows that every matrix containing the columns of the \(k \times k\) identity matrix admits a good regular LD-preconditioner.
\begin{lemma}\label{lm:good_B_norm}
    Let $A \in \RR^{k \times n}$ be a matrix of rank $k$ which contains a $k \times k$ identity submatrix. Then there exists a regular LD-preconditioner $B$ of $A$ with signature $\bigl(\sqrt{e}, 1, 1\bigr)$. Moreover, for rational $A$, such a matrix $B$ can be computed in polynomial time.
\end{lemma}
\begin{proof}
        Let $B$ be an LD-preconditioner of $A$ constructed using \Cref{prop:good_B}.  It is upper triangular with positive diagonal. Thus \(B\) is $1$-integrally regular by taking \(U=I_k\) in
        \Cref{def:norm_LD}. Indeed, since $A$ contains the identity matrix and since $B$ has the signature $\bigl(\sqrt{e}, 1\bigr)$, we have
        \begin{equation*}
            0 < \frac{1}{B_{i i}} = \left( B^{-1}\right)_{i i} \leq \norm{B^{-1} e_i}_2 \leq 1,
        \end{equation*}
        which completes the proof.
\end{proof}

The next theorem constructs, in polynomial time, a rational regular LD-preconditioner for an arbitrary integer matrix of rank $k$. 
\begin{restatable}{theorem}{GoodBConstructiveTh}\label{th:good_B_constructive}
    Let $A \in \ZZ^{k \times n}$ be a matrix of rank $k$. Then there exists a regular LD-preconditioner $B$ of $A$ with signature $\bigl(\sqrt{e}, 1, 1\bigr)$. Moreover, such a matrix $B$ can be computed in polynomial time.
\end{restatable}

\begin{proof}
        Denote the columns of $A$ by $(a_1,\dots,a_n)$ and let $\BC$ be an arbitrary columns base of $A$. Such a base could be constructed by a polynomial-time algorithm (see for example \cite[Section 3.1]{Schrijver}). Denote $\delta = \abs{\det A_{\BC}}$ and $\Delta = \Delta(A)$. 
        Then 
        \[
            \Delta(A_{\BC}^{-1}A) = \frac{\Delta}{\delta},
        \] and $A_{\BC}^{-1}A$ contains an identity submatrix.

        Applying \Cref{lm:good_B_norm} to $A_{\BC}^{-1}A$, we obtain in polynomial time a rational regular LD-preconditioner $B$ of $A_{\BC}^{-1}A$ with signature $\bigl(\sqrt{e}, 1, 1\bigr)$. Hence
        \begin{gather*}
            \abs{\det B} \leq e^{k/2 + o(k)} \cdot \Delta(A_{\BC}^{-1}A) = e^{k/2 + o(k)} \cdot \frac{\Delta}{\delta}, \quad\text{and}\\
            \forall i \in \intint{n}, \quad \norm{B^{-1}(A_{\BC}^{-1} a_i)}_2 \leq 1.
        \end{gather*}

        \noindent Let $B' = A_{\BC} B$. Then, for every $i \in \intint n$,
        \begin{gather*}
            \norm{(B')^{-1} a_i}_2 = \norm{B^{-1}(A_{\BC}^{-1} a_i)}_2 \leq 1,\quad \text{and} \\
            \abs{\det B'} = \abs{\det A_{\BC}} \cdot \abs{\det B} \leq e^{k/2 + o(k)} \cdot \Delta.
        \end{gather*}
        
        \noindent Thus $B'$ is an LD-preconditioner of $A$ with signature $\bigl(\sqrt{e}, 1\bigr)$. Since $B$ is $1$-integrally regular, we can put
        \[
            B = U H,
        \]
        where $U \in \ZZ^{k \times k}$ is unimodular, $H \in \QQ^{k \times k}$ is upper triangular, and the diagonal elements of $H$ are bounded by $1$ from below. Then
        \begin{equation*}
            B'
                = A_{\BC} B = A_{\BC} U H \underset{\substack{A_{\BC} U = U'H_1\\\text{HNF decomposition}}}{=}
                U' H_1 H.
        \end{equation*}
        The matrix $H_1 H$ is upper triangular. Since the diagonal of $H_1$ consists of positive integers and the diagonal of $H$ is bounded by $1$ from below, the diagonal of $H_1 H$ is also bounded by $1$ from below.
        Therefore, $B'$ is the required regular LD-preconditioner of $A$ with signature $\bigl(\sqrt{e}, 1, 1\bigr)$. 
\end{proof}

\section{Enumeration in Parallelepipeds and Ellipsoids via Integrally Regular Matrices}\label{sec:enum_precond}

As already noted earlier, an important part of the dynamic programming algorithm required to prove the main result \Cref{th:KomlosILP} is an algorithm for enumerating integer points in ellipsoids. The algorithm is based on covering the ellipsoid by parallelepipeds and enumerating integer points inside them. The first lemma connects the known \Cref{cor:ParQEnum} for enumeration in parallelepipeds with introduced notion of \emph{integrally regular matrices}.

\begin{proposition}\label{prop:LD_enum_par}
    Let $B \in \QQ^{k \times k}$ be a $\gamma$-integrally regular matrix. Then, for every $t \in \QQ^k$, the points of
    \[
        \bigl(t + B \cdot [0,1)^k\bigr) \cap \ZZ^k
    \]
    can be enumerated using
    \begin{equation*}
        O(k^2) \cdot \left(1 + \frac{1}{\gamma}\right)^k \cdot \abs{\det B} \quad\text{operations.}
    \end{equation*}
\end{proposition}
\begin{proof}
        By the definition of a $\gamma$-integrally regular matrix, there exists a unimodular matrix $U \in \ZZ^{k \times k}$ such that $H = U B$ is upper triangular with diagonal elements bounded by $\gamma$ from below. By \Cref{sec:HNF}, such a matrix can be computed by a polynomial-time algorithm. Denoting $t' = Ut$, the map $x \to U x$ is a bijective map from $\bigl(t + B \cdot [0,1)^k\bigr) \cap \ZZ^k$ to $\bigl(t' + H \cdot [0,1)^k\bigr) \cap \ZZ^k$. By \Cref{cor:ParQEnum}, the points of $\bigl(t' + H \cdot [0,1)^k\bigr) \cap \ZZ^k$ can be enumerated using
        \begin{equation*}
            O(k^2) \cdot \left(1 + \frac{1}{\gamma}\right)^k \cdot \abs{\det H} \quad\text{operations.}
        \end{equation*}
        Since $U$ is unimodular, we have $\abs{\det H} = \abs{\det B}$, which finishes the proof.
        
\end{proof}

For our purposes, we need an algorithm to enumerate integer points in ellipsoids. However, not every ellipsoid is convenient for enumeration. For the sake of clarity, we provide the following definition.
\begin{definition}\label{def:enumerable_matrix}
    Let $B \in \QQ^{k \times k}$ be a non-degenerate matrix. We say that $B$ is \emph{$\alpha$-enumerable} if, for every translation vector $t \in \QQ^k$ and rational radius $r \geq \sqrt{k}$, the integer points inside the ellipsoid $\bigl(t + r B \cdot \BB_2^k \bigr)$ can be enumerated in
    \begin{equation*}
        \alpha^{k + o(k)} \cdot \left(\frac{r}{\sqrt{k}}\right)^k \cdot \abs{\det B}
    \end{equation*}
    arithmetic operations with rational numbers whose bit-encoding size is bounded by the bit-encoding sizes of $B$, $t$, and $r$.
\end{definition}

The following \Cref{prop:LD_enum_Ell} states that $\gamma$-integrally regular matrices are $\alpha$-enumerable for some fixed $\alpha$.
\begin{proposition}\label{prop:LD_enum_Ell}
    Let $B \in \QQ^{k \times k}$ be a $\gamma$-integrally regular matrix. Then $B$ is $\alpha$-enumerable with
    \begin{equation*}
        \alpha = \left( 1 + \frac{1}{\gamma} \right) \cdot \bigl(3 \sqrt{2 \pi e}\bigr).
    \end{equation*}
\end{proposition}
\begin{proof}
    The following simple inequalities can be used to estimate the number of lattice points inside $r \BB_2^k$:
    \begin{equation*}\label{eq:latp_in_ball}
        \vol \left( \left(r - \frac{\sqrt{k}}{2}\right) \cdot \BB_2^k \right) \leq \abs{r \BB_2^k \cap \ZZ^k} \leq \vol \left( \left(r+\frac{\sqrt{k}}{2}\right) \cdot \BB_2^k \right).
    \end{equation*}
    Using $r \geq \sqrt{k}$ from \Cref{def:enumerable_matrix} of $\alpha$-enumerable matrix, we obtain
    \begin{multline}\label{eq:ints_in_ball}
        \abs{r \BB_2^k \cap \ZZ^k} \leq \vol (\BB_2^k) \cdot \left(r + \frac{\sqrt{k}}{2}\right)^k \leq \vol (\BB_2^k) \cdot r^k \cdot (3/2)^k = \\
        = \left( \frac{r}{\sqrt{k}} \right)^k \cdot (2 \pi e)^{k/2} \cdot (3/2)^{k + o(k)} = \left( \frac{r}{\sqrt{k}} \right)^k \cdot \left( \frac{3 \sqrt{\pi e}}{\sqrt{2}} \right)^{k + o(k)}.
    \end{multline}
    Note that the integer points inside $r \BB_2^k$ can be enumerated with the number of operations coinciding with \eqref{eq:ints_in_ball}. 
    

    Now we construct a cover of $r \BB_2^k$ by integer translates of $[0,1)^k$ in the following way. For each integer point inside $r \BB_2^k$, we add $2^k$ translates of $[0,1)^k$ with vertices incident to this point. We claim that this collection covers $r \BB_2^k$. Indeed, assume that there exists an uncovered point $x \in r \BB_2^k$. Without loss of generality, we can assume that $x \geq \BZero$. Since $\norm{\lfloor x \rfloor}_2 \leq \norm{x}_2$, the point $\lfloor x \rfloor$ also lies inside $r \BB_2^k$. But then $x$ is covered by a translated copy of $[0,1)^k$ incident to $\lfloor x \rfloor$, which proves the claim.

    Thus, we have effectively constructed a family of integer translates of $[0,1)^k$ that covers $r \BB_2^k$. By \eqref{eq:ints_in_ball}, this family can be constructed in
    \begin{equation}\label{eq:ball_cube_cover}
        \left( \frac{r}{\sqrt{k}} \right)^k \cdot \left( 3 \sqrt{2 \pi e} \right)^{k + o(k)} \quad\text{operations.}
    \end{equation}

    Finally, we explain how to enumerate integer points inside $t + r B \cdot \BB_2^k$. Since $r \BB_2^k$ is covered by integer translates of $[0,1)^k$, the ellipsoid $t + r B \cdot \BB_2^k$ is covered by integer translates of $t + B \cdot [0,1)^k$; their number is bounded by \eqref{eq:ball_cube_cover}. The enumeration of integer points in each translate is handled by \Cref{prop:LD_enum_par}. The total enumeration complexity is the product of these two bounds. It satisfies the conditions of \Cref{def:enumerable_matrix} for an $\alpha$-enumerable matrix, with $\alpha$ as in the statement of the theorem.
\end{proof}


\section{Proof of \Cref{th:KomlosILP}}\label{sec:proof:KomlosILP}

Recall the statement of \Cref{th:KomlosILP}:
\ThKomlosILP*

\noindent To prove \Cref{th:KomlosILP}, we first construct the dynamic programming algorithm described in \Cref{sec:DP}, and then in \Cref{sec:DP_param_setting} we instantiate its parameters using the construction of a rational regular LD-preconditioner from \Cref{sec:precond}.


\subsection{Dynamic program}\label{sec:DP}

In the seminal work \cite{DiscConvILP}, K.~Jansen \& L.~Rohwedder introduced a new class of dynamic programming algorithms for ILP problems, which utilizes results from discrepancy theory and fast algorithms for tropical and Boolean convolutions on sequences. Our dynamic programming algorithm follows the same pattern but exhibits significant differences, as it employs matrix preconditioning (see \Cref{sec:precond}), enumeration of integer points in ellipsoids (see \Cref{sec:enum_precond}), and handles a more general Boolean convolution problem (see \Cref{sec:bool_conv}). Unfortunately, we do not know any accelerated algorithm for solving the tropical variant of the generalized convolution arising in the optimization problem.

\begin{restatable}{theorem}{DPTh}\label{th:DP}
    Consider the \eqref{ILP-SF} problem. Let $\rho \in \ZZ_{>0}$ be such that $\norm{z^*}_1 \leq (6/5)^\rho$, for some optimal integer solution $z^*$ of the problem. Assume in addition that the following hold for a given nondegenerate matrix \(B \in \QQ^{k \times k}\):
    \begin{enumerate}
        \item There exists a constant $\alpha$ such that $B$ is $\alpha$-enumerable, according to \Cref{def:enumerable_matrix};
        
    
        \item Let $M = B^{-1} \cdot A$. We assume that $\eta \in \ZZ_{\geq 1}$ is a given upper bound on $\herdisc_2(M)$.
    \end{enumerate}
    
    \noindent Under these assumptions on $B$, and denoting $\delta = \abs{\det B}$, the problem can be solved in
    $$
        \rho \cdot \alpha^{2 k + o(k)} \cdot \left(\frac{4 \eta}{\sqrt{k}}\right)^{2k} \cdot \delta^2 + \poly(\inputsize)
    $$ arithmetic operations. The feasibility variant of the problem can be solved in
    $$
        \rho \cdot \alpha^{k + o(k)} \cdot \left(\frac{4 \eta}{\sqrt{k}}\right)^{k} \cdot \delta \cdot \poly(\inputsize,\log \eta)
    \quad \text{bit operations.}
    $$ 
\end{restatable}

\begin{proof}
Set
\[
M = B^{-1}A
\qquad\text{and}\qquad
v = B^{-1}b.
\]
Then $Ax=b$ is equivalent to
\[
Mx=v,\qquad x\in\ZZ_{\geq 0}^n.
\]
Moreover, any feasible right-hand side $u$ of a system $Mx=u$ belongs to $B^{-1} \cdot \ZZ^k$, because
\[
u = Mx = B^{-1}Ax \in B^{-1} \cdot \ZZ^k.
\]

\proofsubparagraph*{Subproblems.}
For $i \in \{0,\dots,\rho\}$ and $u \in B^{-1} \cdot \ZZ^k$, let $\mathcal P(i,u)$ be the problem
\begin{gather*}
    c^\top x \to \max \\
    \begin{cases}
        Mx = u,\\
        \norm{x}_1 \leq (6/5)^i,\\
        x \in \ZZ_{\geq 0}^n.
    \end{cases}
\end{gather*}
Denote its optimum value by $\operatorname{OPT}(i,u)$, and set
$\operatorname{OPT}(i,u)=-\infty$ if $\mathcal P(i,u)$ is infeasible.
Since some optimal solution of $Mx=v$ has $\ell_1$-norm at most $(6/5)^\rho$, the optimum of $\mathcal P(\rho,v)$ is equal to the optimum of the original problem.

The idea is a divide-and-conquer dynamic program: a solution of $\mathcal P(i,u)$ will be split into two smaller solutions corresponding to level $i-1$ with $u', u'' \approx u/2$.

\proofsubparagraph*{Single splitting step.}
Fix $i\geq 1$ and let $x$ be a feasible solution of $\mathcal P(i,u)$.

\noindent If $\norm{x}_1>1$, then by \Cref{lm:disc_rounding} there exists $z\in\ZZ_{\geq 0}^n$ such that
\begin{gather*}
    0\leq z \leq x,\qquad
\norm{M(z-x/2)}_2 \leq 2\eta, \quad \text{and}\\
\frac{1}{6}\norm{x}_1 \leq \norm{z}_1 \leq \frac{5}{6}\norm{x}_1.
\end{gather*}
Hence
\begin{gather*}
    \norm{z}_1 \leq \frac{5}{6}\norm{x}_1 \leq \frac{5}{6}(6/5)^i = (6/5)^{i-1}, \quad \text{and also}\\
    \norm{x-z}_1 = \norm{x}_1 - \norm{z}_1
\leq \frac{5}{6}\norm{x}_1
\leq (6/5)^{i-1}.
\end{gather*}

\noindent If $\norm{x}_1\leq 1$, then the first part of \Cref{lm:disc_rounding} gives a vector
$z\in\ZZ_{\geq 0}^n$ with
\[
0\leq z\leq x,
\qquad
\norm{M(z-x/2)}_2 \leq \eta \leq 2\eta.
\]
In this case,
\[
\norm{z}_1 \leq \norm{x}_1 \leq 1 \leq (6/5)^{i-1},
\qquad
\norm{x-z}_1 \leq \norm{x}_1 \leq 1 \leq (6/5)^{i-1}.
\]

\noindent Thus, in both cases, we obtain $z$ such that
\[
0\leq z\leq x,\qquad
\norm{M(z-x/2)}_2 \leq 2\eta,
\qquad
\norm{z}_1,\ \norm{x-z}_1 \leq (6/5)^{i-1}.
\]

Now, we set
\[
u' = Mz,
\qquad
u'' = M(x-z)=u-u'.
\]
Then $u',u'' \in B^{-1} \cdot \ZZ^k$, and
\begin{gather*}
    \norm{u'-u/2}_2
=
\norm{Mz-u/2}_2
=
\norm{M(z-x/2)}_2
\leq 2\eta,\\
\norm{u''-u/2}_2
=
\norm{M(x-z)-u/2}_2
=
\norm{M(x/2-z)}_2
=
\norm{Mz-u/2}_2
\leq 2\eta;\\
\text{Thus}\quad u', u''\in u/2+2\eta \cdot \BB_2^k.
\end{gather*}

Since $z$ and $x-z$ satisfy the required norm bound, they are feasible for
$\mathcal P(i-1,u')$ and $\mathcal P(i-1,u'')$, respectively. Therefore, if $x$ is optimal for $\mathcal P(i,u)$, then
\[
\operatorname{OPT}(i,u)
=
c^\top x
=
c^\top z + c^\top(x-z)
\leq
\operatorname{OPT}(i-1,u') + \operatorname{OPT}(i-1,u'').
\]
Conversely, if $x'$ and $x''$ satisfy
\begin{gather*}
    Mx' = u',
\qquad
Mx'' = u'', \\
\text{then}\quad M(x'+x'')=u;
\end{gather*}
Moreover, if
\begin{gather*}
c^\top x'\geq \operatorname{OPT}(i-1,u') \qquad\text{and}\qquad c^\top x''\geq \operatorname{OPT}(i-1,u''),\\
\text{then}\quad c^\top(x'+x'')\geq \operatorname{OPT}(i,u).
\end{gather*}
Hence, to obtain a value at least $\operatorname{OPT}(i,u)$, it is enough to know solutions whose objective values are at least $\operatorname{OPT}(i-1,u')$ and $\operatorname{OPT}(i-1,u'')$; the corresponding witnesses do not need to satisfy the norm bound.

\proofsubparagraph*{Relevant state space.}
For $j\in\{0,\dots,\rho\}$, define
\[
\mathcal U_j
=
\left(\frac{v}{2^j}+4\eta \cdot \BB_2^k\right)\cap B^{-1} \cdot \ZZ^k.
\]
If $u\in\mathcal U_j$, then every child state $u'$ produced by the splitting step belongs to $\mathcal U_{j+1}$, because
\[
\begin{aligned}
u' &\in u/2+2\eta \cdot \BB_2^k \subseteq
\frac{1}{2}\left(\frac{v}{2^j}+4\eta \cdot \BB_2^k\right)+2\eta \cdot \BB_2^k \\
&=
\frac{v}{2^{j+1}}+2\eta \cdot \BB_2^k+2\eta \cdot \BB_2^k \\
&=
\frac{v}{2^{j+1}}+4\eta \cdot \BB_2^k.
\end{aligned}
\]
Thus, if $u\in\mathcal U_j$, then $u'\in\mathcal U_{j+1}$.

We will store a state $u\in\mathcal U_j$ by its integer representative $b' \in \ZZ^k$, given by $B^{-1}b' = u$.
Accordingly, define
\[
\mathcal B_j
=
\left\{\,b'\in\ZZ^k : B^{-1}b' \in \mathcal U_j\,\right\}
=
\ZZ^k \cap \left(\frac{b}{2^j} + 4\eta B \cdot \BB_2^k\right).
\]
By the assumption on $B$, according to \Cref{def:enumerable_matrix}, each set $\mathcal B_j$ can be enumerated in
$\alpha^{k + o(k)} \cdot \bigl(4\eta/\sqrt{k}\bigr)^k\cdot \delta$ operations. In particular,
\[
    \abs{\mathcal B_j} = \alpha^{ k + o(k)} \cdot \left(\frac{4\eta}{\sqrt{k}}\right)^k\cdot \delta.
\]

\proofsubparagraph*{Dynamic programming tables.}
For $i\in\{0,\dots,\rho\}$ and $b'\in\mathcal B_{\rho-i}$, let
\[
u = B^{-1}b'.
\]
In the feasibility and optimization variant we store
\begin{equation*}
    \DP_{\mathrm{feas}}[i,b'] \in \{0,1\} \quad\text{and}\quad \DP[i,b'] \in \ZZ \cup \{-\infty\}, \quad \text{respectively}.
\end{equation*}
In both the optimization and feasibility variants, the $\DP$ entry corresponds to a solution of $Mx=u$. This solution need not be feasible for \(\mathcal P(i,u)\); it may violate the norm bound. However, in the optimization variant we maintain the invariant that its objective value is at least $\operatorname{OPT}(i,u)$.
Precisely, for every state \((i,b')\), we maintain the following invariants:
\begin{itemize}
    \item If \(\mathcal P(i,u)\) is feasible, then
    \begin{align}
        \DP_{\mathrm{feas}}[i,b'] &= 1,
        \notag\\
        \DP[i,b'] &\ge \operatorname{OPT}(i,u).
        \label{eq:first_inv}\tag{INVARIANT-1}
    \end{align}

    \item If a table entry certifies feasibility, then a corresponding solution exists:
    \begin{align}
        \DP_{\mathrm{feas}}[i,b'] = 1
        & \;\Longrightarrow\;
        \exists x\in\ZZ_{\ge 0}^n \text{ s.t. } Mx=u,
        \notag\\
        \DP[i,b'] > -\infty
        & \;\Longrightarrow\;
        \exists x\in\ZZ_{\ge 0}^n \text{ s.t. } Mx=u
        \text{ and } c^\top x = \DP[i,b'].
        \label{eq:second_inv}\tag{INVARIANT-2}
    \end{align}
\end{itemize}

\proofsubparagraph*{Initialization.}
For $i=0$, the bound $\norm{x}_1\leq 1$ leaves only the possibilities $x\in\{0,e_1,\dots,e_n\}$.
We inspect all these vectors and initialize exactly the non-default entries of
$\DP[0,\cdot]$ and $\DP_{\mathrm{feas}}[0,\cdot]$; every remaining state keeps the default value
$-\infty$ or $0$, respectively. This takes only $\poly(\inputsize)$ operations.

\proofsubparagraph*{Transitions.}
For $i\geq 1$ and $b'\in\mathcal B_{\rho-i}$, define
\begin{equation}\label{eq:DP_formula}\tag{Opt-Trans}
\DP[i,b']
=
\max_{\substack{b''\in\mathcal B_{\rho-i+1}\\ b'-b''\in\mathcal B_{\rho-i+1}}}
\Bigl(
\DP[i-1,b''] + \DP[i-1,b'-b'']
\Bigr),
\end{equation}
with the convention that the maximum over the empty set is $-\infty$, and
\begin{equation}\label{eq:DP_feasibility}\tag{Feasibility-Trans}
\DP_{\mathrm{feas}}[i,b']
=
\bigvee_{\substack{b''\in\mathcal B_{\rho-i+1}\\ b'-b''\in\mathcal B_{\rho-i+1}}}
\Bigl(
\DP_{\mathrm{feas}}[i-1,b''] \wedge
\DP_{\mathrm{feas}}[i-1,b'-b'']
\Bigr).
\end{equation}

\proofsubparagraph*{Correctness and Running Time.}

The algorithm's correctness is equivalent to validity of \eqref{eq:first_inv} and \eqref{eq:second_inv}, and it follows from the definitions of subproblems and $\DP[i,b']$. The running time depends on how we evaluate formulas \eqref{eq:DP_formula} and \eqref{eq:DP_feasibility}. We evaluate formula \eqref{eq:DP_formula} using a trivial algorithm. For the formula \eqref{eq:DP_feasibility}, we use the fast sparse Boolean convolution algorithm from \Cref{lm:bool_conv}. Due to space constraints, a detailed proof of correctness and a running-time analysis are deferred to \Cref{th:proof:DP}.
    
\end{proof}

\subsection{Adjusting the Parameters}\label{sec:DP_param_setting}

To prove \Cref{th:KomlosILP}, we will use \Cref{th:DP} with appropriately chosen parameters $\rho$, $\eta$, $B$, and $\delta$.
The parameter $\rho$ can be chosen using a standard argument based on the proximity of the LP and ILP solutions. The details are encapsulated in the following lemma, whose proof is deferred to \Cref{lm:proof:rho_estimate}. Due to this lemma, we can assume that $\rho = O\bigl(\log(k\Delta)\bigr)$.
\begin{restatable}{lemma}{RhoEstimateLm}\label{lm:rho_estimate}
    Any instance of the \eqref{ILP-SF} problem can be transformed to an equivalent instance with the following property: if the problem is feasible and bounded, then there exists an optimal solution $z^*$ such that $\norm{z^*}_1 = (k \cdot \Delta)^{O(1)}$.
\end{restatable}

Applying \Cref{th:good_B_constructive} to $A$, we obtain a rational regular LD-preconditioner $B$ of $A$ with signature $\bigl(\sqrt{e}, 1, 1\bigr)$. By \Cref{prop:LD_enum_Ell} and according to \Cref{def:enumerable_matrix}, the matrix $B$ is $\alpha$-enumerable with $\alpha = 6 \sqrt{2 \pi e}$.


By the definition of a regular LD-preconditioner with signature $\bigl(\sqrt{e}, 1, 1\bigr)$,
\[
    \delta = \abs{\det B} \leq e^{k/2 + o(k)} \cdot \Delta.
\]
And by \Cref{prop:LD_herdisc}, we can set
\[
    \eta = \lceil \sqrt{k} \rceil.
\]
Now the proof follows from \Cref{th:DP}.



\bibliography{grib_biblio}

\appendix

\section{Omitted Proofs}

\ifarxive

        \subsection{Proof of \Cref{lm:approx_LJ_ell}}\label{lm:proof:approx_LJ_ell}

        Recall the statement of \Cref{lm:approx_LJ_ell}.

        \ApproxLJEllLm*

        \begin{proof}
            A feasible solution $X$ of the optimization problem $\min \bigl\{f_0(X) \colon f_i(X) \leq 0\bigr\}$ is called \emph{$\alpha$-optimal} if $f_0(X) \leq f^* + \alpha$, where $f^*$ is the optimal value of the problem. Using the ellipsoid method \cite{ApplicationsOfEllipsoidMethod}, one can compute an $\alpha$-optimal rational feasible solution $\widehat W$ of \eqref{eq:minell_primal} in time polynomial in the encoding length of $\AC$ and in $\log \alpha^{-1}$. Taking $\alpha = 2\varepsilon$, we obtain
            \[
                -\ln \det \widehat W \leq -\ln \det W + 2\varepsilon,
            \]
            and therefore
            \[
                \det\widehat W^{1/2} \geq e^{-\varepsilon} \cdot \det W^{1/2}.
            \]

            Compute the $L^\top D L$ decomposition
            \[
                \widehat W = L^\top D L,
            \]
            where $L$ is rational upper triangular with unit diagonal and $D$ is diagonal with positive rational entries. This decomposition can be found in polynomial time. Let $\sqrt{D}$ denote the diagonal matrix obtained by taking entrywise square roots of the diagonal entries of $D$.

            For each $i \in \intint k$, compute a rational number $d'_i$ such that
            \[
                e^{-\varepsilon/k}\sqrt{D_{ii}} \leq d'_i \leq \sqrt{D_{ii}}.
            \]
            Let $D' = \diag(d'_1,\dots,d'_k)$. Then, $D'^2 \leq D$ entrywise and
            \[
                \det D' \geq e^{-\varepsilon} \cdot \det \sqrt{D}.
            \]

            \noindent Set
            \[
                C = D' L.
            \]
            Clearly, $C$ is upper triangular with positive diagonal. For every $a \in \AC$, we have
            \[
                a^\top C^\top C a = a^\top L^\top D'^2 L a \leq a^\top L^\top D L a = a^\top \widehat W a \leq 1,
            \]
            because $\widehat W$ is feasible for \eqref{eq:minell_primal}. Using $\det L = 1$, we obtain
            \[
                \det C
                = \det D'
                \geq e^{-\varepsilon} \cdot \det D^{1/2}
                = e^{-\varepsilon} \cdot \det \widehat W^{1/2}
                \geq e^{-2\varepsilon} \cdot \det W^{1/2}.
            \]
        \end{proof}

\fi

        \subsection{Proof of \Cref{lm:ParQEnum}}\label{lm:proof:ParQEnum}

        Recall the statement of \Cref{lm:ParQEnum}.

        \ParQEnumLm*

       \begin{proof}

We enumerate the points of $\PC \cap \ZZ^n$ from the last coordinate to the first. Take
$y \in \PC$, so
\[
y = r + Ht,
\qquad
t \in [0,1)^n.
\]
For $k = n, n-1, \dots, 1$, assume that $y_{k+1}, \dots, y_n$ have already been fixed (for $k=n$ this means that nothing is fixed yet). The corresponding values $t_{k+1}, \dots, t_n$ (such that equalities $y_i = H_{i*} t$ for $y_{k+1}, \dots, y_n$ are satisfied) are known as well. We know that
\[
y_k = r_k + \sum_{j=k}^n H_{k j} t_j = r_k + \sum_{j = k+1}^n H_{k j} t_j + H_{k k} t_k,
\qquad t_k \in [0,1)
\]
Set
\[
\tau_k = r_k + \sum_{j = k+1}^n H_{k j} t_j
\] 
Hence, the admissible values of $y_k$ are exactly the integers in the interval
\[
\tau_k + H_{k k}\cdot [0,1) = [\tau_k, \tau_k + H_{k k}).
\]
This interval has length $H_{k k}$, hence it contains at least $\lfloor H_{k k} \rfloor$ and at most $\lceil H_{k k} \rceil$ integers. These integers can be listed simply by scanning this interval, so this step takes $O\bigl(\lceil H_{k k} \rceil\bigr)$ time. After choosing $y_k$, we recover
\[
t_k = \frac{y_k - \tau_k}{H_{k k}}.
\]

\noindent Thus, this procedure defines a search tree: at step $k$, every node has between $\lfloor H_{k k} \rfloor$ and $\lceil H_{k k} \rceil$ children, and the leaves are exactly the points of $\PC \cap \ZZ^n$. Therefore,
\[
\prod_{i=1}^n \lfloor H_{i i} \rfloor
\le
\abs{\PC \cap \ZZ^n}
\le
\prod_{i=1}^n \lceil H_{i i} \rceil.
\]

\noindent For each such node on level $k$, we spend $O(n)\cdot \lceil H_{kk}\rceil$ operations. 
Hence the total work at level $k$ is $O(n)\cdot \prod_{i=1}^n \lceil H_{ii}\rceil$.
Since there are $n$ levels, the overall number of operations coincides with the stated complexity bound.
\end{proof}

        \subsection{Proof of \Cref{lm:c_opt}}\label{lm:proof:c_opt}

            Recall the statement of \Cref{lm:c_opt}.
            \COptLm*

            \begin{proof}
              Let $c$ be an optimal solution, chosen so that $\sum_{i=1}^n c_i^2$ is minimal among all optimal solutions. Suppose $c_l \neq c_s$ for some $l$ and $s$. Define $c'$ to be the vector obtained from $c$ by replacing both $c_l$ and $c_s$ with $\frac{c_l+c_s}{2}$. Then,
               \begin{multline*}
                   \sum\limits_{\JC \in \binom{\intint n}{k}} \prod\limits_{j \in \JC} c_{j}' - \sum\limits_{\JC \in \binom{\intint n}{k}} \prod\limits_{j \in \JC} c_j = \\
                   \sum\limits_{\JC \in \binom{\intint n \setminus \{ l, s \} }{k - 2}}\left(c'_l c'_s \prod\limits_{j \in \JC} c'_{j} - c_l c_s \prod\limits_{j \in \JC} c_j\right) = \\
                   \left(\biggl(\frac{c_l+c_s}{2} \biggr)^2 - c_l c_s \right) \cdot \sum\limits_{\JC \in \binom{\intint n \setminus \{ l, s \} }{k - 2}}\prod\limits_{j \in \JC} c_j \geq 0.
               \end{multline*}
               Hence the objective value at $c'$ is no smaller than at $c$. Since $c$ is optimal, it follows that $c'$ is optimal as well.
               On the other hand,
                \[
                \sum_{i=1}^n (c'_i)^2 < \sum_{i=1}^n c_i^2,
                \]
                which contradicts the choice of $c$.
            \end{proof}

        \subsection{Proof of \Cref{lm:bool_conv}}\label{lm:proof:bool_conv}

        Recall the statement of \Cref{lm:bool_conv}.
        \BoolConvLm*

        \begin{proof}
            Construct the polynomials $f_{\alpha},f_{\beta} \in \ZZ[x_1, \dots,x_k]$ in the following way. For each $a \in \supp(\alpha)$, we include the monomial $\xB^{a}$ in $f_{\alpha}$. We do the same for each $b \in \supp(\beta)$ to construct $f_{\beta}$. To make the exponents in both polynomials non-negative, we additionally multiply them by $\xB^{L \cdot \BUnit}$. Denoting $g = f_{\alpha} \cdot f_{\beta}$, we have
            \begin{equation*}
                g[\xB^{a+2L \cdot \BUnit}] = \sum\limits_{\substack{b \in \supp(\alpha)\\(a-b) \in \supp(\beta)}} \alpha_b \cdot \beta_{a-b}.
            \end{equation*}
            Since $\gamma_a = 1$ iff $g[\xB^{a+2L \cdot \BUnit}] \neq 0$, the computation of $\gamma$ reduces to multiplying $f_{\alpha}$ and $f_{\beta}$.
        
            By construction, $\max\left\{\norm{f_\alpha}_\infty,\norm{f_\beta}_\infty\right\} = 1$ and $\max_i \deg_{x_i} (g) \leq 4 L$. Therefore, by \Cref{th:poly_mult}, the function $\gamma$ can be computed in
            \begin{equation*}
                O\left(T \cdot (k \log L ) \cdot \polylog(1/\varepsilon) \right)
            \end{equation*}
            bit operations with probability $1-\varepsilon$. Taking $\varepsilon = 1/k^{O(1)}$, we can guarantee the correctness with high probability.
        \end{proof}
        \begin{remark}\label{rem:bool_conv_det}
    There is also a deterministic analogue of \Cref{lm:bool_conv}. Let $q = 4L + 1$ and encode each vector $a = (a_1,\dots,a_k) \in [-L,L]^k \cap \ZZ^k$ by
    \[
        \phi(a) = \sum_{i = 1}^{k} (a_i + L) q^{i-1}.
    \]
    Then for any $a,b \in [-L,L]^k \cap \ZZ^k$, the base-$q$ expansion of $\phi(a) + \phi(b)$ is exactly $a + b + 2L \cdot \BUnit$, because every digit belongs to $\{0,\dots,4L\}$ and hence no carries occur. Therefore, \Cref{prob:bool_conv} reduces to sparse nonnegative convolution of two one-dimensional $0/1$ vectors of length $q^k$.

    Applying the algorithm of Bringmann, Fischer, and Nakos~\cite[Theorem~1]{DetLasVegasSparseNonnegativeConv}, we obtain a running time
    \[
        O\bigl(T \cdot \polylog(q^k\cdot L)\bigr) = O\bigl(T \cdot \poly(k,\log L)\bigr)
    \]
    in the Word-RAM model from \Cref{assumptions_subs}. Hence \Cref{lm:bool_conv} admits a deterministic version with complexity $O\bigl(T \cdot \poly(k,\log L)\bigr)$.
    
    \noindent A similar construction for Boolean convolution is used by Jansen and Rohwedder \cite{DiscConvILP}.
\end{remark}

    \subsection{Proof of \Cref{lm:rho_estimate}}\label{lm:proof:rho_estimate}

            Recall the statement of \Cref{lm:rho_estimate}.
            \RhoEstimateLm*
            
            \begin{proof}
              Let \(v\) be an optimal vertex solution of the relaxation; such a solution can be computed in polynomial time. By \cite[Corollary~2]{OnCanonicalProblems_Grib}, there exists an optimal solution \(\hat z\) to the original problem such that
            \[
            \|v-\hat z\|_1 \le \chi,
            \qquad
            \chi = O\!\left(k^2 \cdot \Delta \cdot \sqrt[k]{\Delta}\right).
            \]
            Following \cite{DiscConvILP}, we apply the change of variables \(x' = x-y\), where
            \[
            y_i = \max\{0,\lceil v_i\rceil-\chi\}.
            \]
            Then \(\hat z \ge y\). Moreover, since \(v\) has at most \(k\) nonzero components, it follows that
            \[
            \|\hat z-y\|_1 = O(k\chi).
            \]
            Under this substitution, the original problem becomes an equivalent instance of \eqref{ILP-SF} with optimal solution \(z^*=\hat z-y\). Therefore,
            \[
            \|z^*\|_1 = O(k\chi) = (k\Delta)^{O(1)}.
            \]
            \end{proof}

    \subsection{Proof of the DP-algorithm's Correctness and Running Time Analysis (\Cref{th:DP}) }\label{th:proof:DP}

    \proofsubparagraph*{Correctness.}
We prove the invariants by induction on $i$. For $i=0$, they hold by the exact initialization. Assume that they hold on level $i-1$, and consider a state
$b'\in\mathcal B_{\rho-i}$ with $u=B^{-1}b'$.

Suppose first that $\mathcal P(i,u)$ is feasible, and let $x$ be an optimal solution.
By the splitting step, there exist $u',u''\in u/2+2\eta \cdot \BB_2^k$ with
$u=u'+u''$, such that $\mathcal P(i-1,u')$ and $\mathcal P(i-1,u'')$ are feasible and
\[
\operatorname{OPT}(i,u)
\leq
\operatorname{OPT}(i-1,u')+\operatorname{OPT}(i-1,u'').
\]
Since $u\in\mathcal U_{\rho-i}$, as shown above
$u',u''\in\mathcal U_{\rho-i+1}$.
Let
\[
b'' = Bu' = Az.
\]
Then $b'',\,b'-b''\in\mathcal B_{\rho-i+1}$, and by the induction hypothesis,
\[
\DP[i-1,b''] \geq \operatorname{OPT}(i-1,u'),
\qquad
\DP[i-1,b'-b''] \geq \operatorname{OPT}(i-1,u''),
\]
as well as
\[
\DP_{\mathrm{feas}}[i-1,b''] =
\DP_{\mathrm{feas}}[i-1,b'-b''] = 1.
\]
Therefore, \eqref{eq:DP_formula} and \eqref{eq:DP_feasibility} imply
\[
\DP[i,b'] \geq \operatorname{OPT}(i,u),
\qquad
\DP_{\mathrm{feas}}[i,b'] = 1.
\]
So \eqref{eq:first_inv} is preserved.

\smallskip

\noindent To show that \eqref{eq:second_inv} is also preserved, assume first that
\[
\DP_{\mathrm{feas}}[i,b'] = 1.
\]
By \eqref{eq:DP_feasibility}, there exists
$b''\in\mathcal B_{\rho-i+1}$ with $b'-b''\in\mathcal B_{\rho-i+1}$ such that
\[
\DP_{\mathrm{feas}}[i-1,b''] =
\DP_{\mathrm{feas}}[i-1,b'-b''] = 1.
\]
By the induction hypothesis, there exist vectors
$x',x''\in\ZZ_{\geq 0}^n$ satisfying
\[
Mx' = B^{-1}b'',
\qquad
Mx'' = B^{-1}(b'-b'').
\]
Hence
\[
M(x'+x'') = B^{-1}b'' + B^{-1}(b'-b'') = B^{-1}b' = u.
\]
Thus $\DP_{\mathrm{feas}}[i,b']=1$ certifies the existence of a feasible solution of $Mx=u$. The same argument applies to $\DP[i,b'] > -\infty$. So \eqref{eq:second_inv} is preserved.

\proofsubparagraph*{Conclusion of correctness.}
Now consider the root state $b$, corresponding to $u=v$.

For the optimization problem, \eqref{eq:first_inv} gives
\[
\DP[\rho,b] \geq \operatorname{OPT}(\rho,v).
\]
On the other hand, if $\DP[\rho,b]>-\infty$, then by \eqref{eq:second_inv} there exists
$x\in\ZZ_{\geq 0}^n$ such that
\[
Mx=v
\qquad\text{and}\qquad
c^\top x = \DP[\rho,b].
\]
Hence $\DP[\rho,b]$ cannot exceed the true optimum of $Mx=v$.
Since some optimal solution has $\ell_1$-norm at most $(6/5)^\rho$, the true optimum is exactly $\operatorname{OPT}(\rho,v)$. Therefore,
\[
\DP[\rho,b] = \operatorname{OPT}(\rho,v),
\]
and the optimization algorithm is correct.

For the feasibility problem, if the instance is feasible, then the assumed optimal solution is feasible for $\mathcal P(\rho,v)$, and \eqref{eq:first_inv} yields
\[
\DP_{\mathrm{feas}}[\rho,b] = 1.
\]
Conversely, if
\[
\DP_{\mathrm{feas}}[\rho,b] = 1,
\]
then \eqref{eq:second_inv} gives a vector $x\in\ZZ_{\geq 0}^n$ with $Mx=v$.
Thus
\[
\DP_{\mathrm{feas}}[\rho,b] = 1
\quad\Longleftrightarrow\quad
Mx=v \text{ is feasible},
\]
and the feasibility algorithm is correct as well.

\proofsubparagraph*{Running time.}
As noted above, every set $\mathcal B_j$ has size
\[
    \abs{\mathcal B_j} = \alpha^{k + o(k)} \cdot \left(\frac{4\eta}{\sqrt{k}}\right)^k \cdot \delta.
\]

\noindent For the optimization table, a single level is computed by the trivial quadratic scan in
\eqref{eq:DP_formula}, so it costs
\[
    \alpha^{2 k + o(k)}\cdot \left(\frac{4 \eta}{\sqrt{k}}\right)^{2k} \cdot \delta^2.
\]
arithmetic operations.
Summed over all $\rho$ levels, and including the initialization, this gives
\[
\rho \cdot \alpha^{2 k + o(k)} \cdot \left(\frac{4 \eta}{\sqrt{k}}\right)^{2k} \cdot \delta^2 + \poly(\inputsize)
\]
arithmetic operations.



Applying \Cref{lm:bool_conv}, one level of the feasibility dynamic program can be computed in
\[
    \alpha^{ k + o(k)} \cdot \left(\frac{4 \eta}{\sqrt{k}}\right)^k \cdot \delta \cdot \poly(\inputsize, \log \eta)
\]
bit operations.
Over all $\rho$ levels, the total running time is
\[
\rho \cdot \alpha^{k + o(k)} \cdot \left(\frac{4 \eta}{\sqrt{k}}\right)^k \cdot \delta \cdot \poly(\inputsize, \log \eta)
\]
bit operations.

\end{document}